\documentclass[journal]{IEEEtran}

\usepackage{amsmath,amssymb}
\usepackage{amsthm}
\usepackage{graphicx,color}
\usepackage{booktabs}
\usepackage[linesnumbered,ruled]{algorithm2e}
\def\opt{\mbox{\footnotesize opt}}

\def\even{\mbox{\footnotesize even}}
\def\odd{\mbox{\footnotesize odd}}

\def\app{\mbox{\footnotesize app}}

\def\eE{{\mathbb E}}
\theoremstyle{remark}
\newtheorem{theorem}{Theorem}[section]

\newtheorem{prop}{Proposition}

\newtheorem{remark}{Remark}
\graphicspath{{figures/}} 
\begin{document}
%
\title{Novel Near-Optimal Scalar Quantizers with Exponential Decay Rate and
Global Convergence}

\author{Vijay~Anavangot,~\IEEEmembership{Student Member,~IEEE,}
        and~Animesh~Kumar,~\IEEEmembership{Member,~IEEE}
}

%
\maketitle
\begin{abstract}
  Many modern distributed real-time signal sensing/monitoring systems require
  quantization for efficient signal representation. These distributed sensors
  often have inherent computational and energy limitations. Motivated by this
  concern, we propose a novel quantization scheme called approximate Lloyd-Max
  that is nearly-optimal. Assuming a continuous and finite support probability
  distribution of the source, we show that our quantizer converges to the
  classical Lloyd-Max quantizer with increasing bitrate. We also show that our
  approximate Lloyd-Max quantizer converges exponentially fast with the number
  of iterations. The proposed quantizer is modified to account for a relatively
  new quantization model which has envelope constraints, termed as the envelope
  quantizer. With suitable modifications we show optimality and convergence
  properties for the equivalent approximate envelope quantizer.  We illustrate
  our results using extensive simulations for  different finite support
  distributions on the source.  
\end{abstract}

\begin{IEEEkeywords}
  Quantization (signal), Piecewise linear approximation, Data models,
  Probability distribution, Computational efficiency
\end{IEEEkeywords}

%
\section{Introduction}
\IEEEPARstart{T}{he} widespread deployment of sensors for monitoring systems
(such as pollution, weather) will generate gigantic amounts of (discrete)
signals/data. For efficient storage and communication, data compression methods
such as quantization will play a vital role. In order to account for limited
computation and energy available at these sensor nodes (due to large scale of IoT
devices and mobile sampling), novel quantization algorithms will be
necessary.

Scalar quantization of a signal with known probability distribution
is studied in the well-known works of Lloyd and Max~\cite{Lloyd:1982,Max:1960}.
With the advent of distributed signal processing and in-network
computations~\cite{Predd:2006,Dimakis:2010,Vyavahare:2016} in large scale sensor
deployments, this (locally) optimal quantization scheme is infeasible due to
limited energy, bandwidth and computational power at the terminal sensor nodes.  
The classical Lloyd-Max algorithm requires integral computations in the centroid
(conditional mean) update step. In this work we introduce a nearly-optimal
scalar quantization algorithm, known as Approximate Lloyd-Max (ALM) that
bypasses these computationally complex operations.  We show exponentially fast
convergence of ALM to the near global optima for a class of source distributions
where Lloyd-Max quantizer is globally optimal. Our algorithm uses vectorized
update rules that are governed by linear transforms derived from localized
mean square error optimizations.

%
%
The approximate Lloyd-Max algorithm deals with the mean square error cost
function. We also consider a second cost function, known as envelope constrained
mean square error quantization or shortened as envelope quantizer, which is
relatively new and applicable to domains spanning, environmental monitoring,
protection region database for TV whitespace and others~\cite{Maheshwari:2015}.
The equivalent of ALM in this case is known as Approximate Envelope Quantizer
(AEQ).  We show that the AEQ scheme inherits all the properties of the ALM, by
suitable modifications to accomodate the envelope constraint.
%
%
The classical Lloyd-Max
algorithm is known to have global convergence to a unique local minima, 
under certain restrictive class of cost function and the probaility
distribution~\cite{Sabin:1986,Wu:1992}. Our algorithms give a generalized proof
method to establish the global convergence for the class of continuously
distributed  sources on a finite support. The convergence result is hinged on
the linear level update rule obtained as a result of the cost minimization in a
local neighborhood. The same proof methodology applies to AEQ, where the additional
envelope constraint is handled suitably through level shifting.

The main contributions of this work are summarized below.
\begin{enumerate}\itemsep1pt
  \item Linear approximation based quantization schemes - Approximate Lloyd-Max (ALM)
    and Approximate Envelope Quantizer (AEQ). A vectorized algorithm termed as
    Alternating Between Evens and Odds (ABEO) is proposed, which performs
    parallel computation of all the quantization levels in each iteration.
  \item Convergence of ALM and AEQ is established using linear matrix
    transformations and the Perron-Frobenius
    theory~\cite{PFTheory:2005,Gallager:2012}. We show that both the algorithms
    convege to a unique global minima, at exponential decay rate    convergence.
  \item Near optimality of the approximation based quantization scheme is
    analytically established.
  \item Simulations on source models with finite support are performed to
    characterize the error vs bitrate tradeoff. Experiments also confirms the
    near-optimality of the quantization schemes.
\end{enumerate}

\textit{Key ideas} : The application of function approximations for quantizer
error cost optimization forms the main theme of this work. The approximation
scheme proposed here, simultaneously satisfy computational simplicity and
accuracy of the quantization level computation. The level updates thus obtained,
is represented as a sequence of linear (matrix) transformations that satisfy
row-stochastic property. The structure of these matrices enables us to use ideas
from Perron-Frobenius theory to establish global convergence of the proposed
algorithms. The row-stochastic nature of the vectorized update rule also finds
connections in gossip algorithms and consensus models~\cite{Shah:2009,Huang:2012}.

\section{Related Literature}

Fixed rate optimal scalar quantization with known data distribution and mean
square error cost function, was first studied in the independent works by Lloyd
and Max~\cite{Lloyd:1982, Max:1960}. This iterative scheme, popularly known as
the Lloyd-Max algorithm, minimizes the mean square error using alternate update
of the the decision boundaries and the quantization levels.  Sharma extended the
Lloyd-Max method to a general class of (convex/semiconvex) distortion
measures~\cite{Sharma:1978}. This work employs a combination of dynamic
programming and fast search in order to iteratively update the quantization
levels. An algorithm for quantizer design considering vector data was first
introduced by Linde, Buzo and Gray~\cite{LBG:1980}. Quantizer design based on a
known probabilistic model or on a sequence of long training data is demonstrated
here. Vector quantization follows from an efficient extension of the Lloyd-Max
algorithm to higher dimensions. Another work related to vector quantization,
considers predictive quantizer design hinging on tree based search methods for
the class of Gauss-Markov sources~\cite{Gray:1982}. Ziv proposed a variable rate
universal quantizer for vector data, that achieves the optimal performance
within a constant gap~\cite{Ziv:1985}. Other variants of high dimensional
quantization using lattices and Voronoi tessellations are widely studied in
mathematical literature~\cite{Du:2006,Conway:1982, Conway:1998}. Gray and
Neuhoff have summarized the historical evolution of the quantization schemes,
both scalar and vector cases, in their comprehensive review
paper~\cite{Gray:1998}.  A practical approach to implement quantizers under
limited computational power and memory constraints are addressed by
Gersho~\cite{Gersho:1991}. In this respect, suboptimal and asymptotically
optimal (with respect to the number of quantization levels) schemes are
proposed.  The computational constraints discussed here differ from the envelope
constraints motivated in the current paper.  In the former case the constraints
are due to computational costs at higher dimension, while in this paper the
constraints arise from application specific design requirements.

The convergence analysis of the Lloyd-Max algorithm is widely studied in
literature. Convergence with exponential decay rate to a unique local minima,
under the assumptions of a convex cost function and a log-concave probability
distribution~\cite{Kieffer:1982}. The work by Sabin and Gray explains the
absolute convergence of the Lloyd algorithm and its empirical density
consistency on training data~\cite{Sabin:1986}.  The correspondence by Wu shows
the convergence of the Lloyd method I for continuous, positive density function
defined over a finite interval using the idea of finite state
machines~\cite{Wu:1992}. Another variant of this work explores dynamic
programming based global optimum search methods based on the monotony properties
of the error function~\cite{Wu:1993}.  The authors show a quadratic time
algorithm that converges to the global minima.  A more recent work provides a
linear time accelerated multigrid search algorithm that applies to both
continuous and discrete scalar probability densities~\cite{Koren:2005}. Several
works have studied the quantization problem in relation to the K-means
clustering framework~\cite{Pollard:1982,Bottou:1995,Kanungo:2002}. Pollard
introduced a novel approach to show the consistency theorem for k-means cluster
centers and its relation on data centric quantization. Bottou and Bengio
established the optimality of the K-means algorithm using gradient descent and
fast Newton algorithm~\cite{Bottou:1995}.  High dimensional Voronoi tessellation
under generalized assumptions, such as compact support are covered in some
relatively recent works~\cite{Du:2006,Emelianenko:2008}.

Our contributions differ from the previous literature in the following aspects.
We show an exponentially fast, globally optimal convergence result for the class
of finitely supported quantizers. Also the proposed algorithm is not restricted
to the class of log concave (unimodal) distributions, as assumed in many of the
earlier works. Our method is analytically and computationally efficient, as it
involves only the use of a sequence of linear transformations and convergence
based on the Perron Frobenius theory. The proposed scheme extends to more
generic cost functions and optimization constraints. This is possible since the
update rules are based on local optimization of the cost function.

\textit{Notations} : We use the notation, $X$ to represent a random
variable and $x$ to denote the realization of the random variable. $f_X(x)$
represents the density function of $X$. For ease of exposition, we use the
phrase "quantization for continuous distribution" to indicate "quantization for
continuously distributed sources".

\textit{Organization} :  
We first develop the cost function, optimality criteria and update rule
corresponding to nearly optimal quantizers, ALM and AEQ (see
\ref{sec:opt_env_quant}). In Section \ref{sec:conv_env_quant}, the main
analytical result relating to the global optimality and exponential convergence
of the two proposed algorithms is presented. Section~\ref{sec:sim_results}
discusses the simulation experiments and convergence trends of our algorithms.
%
%
\section{Nearly Optimal Quantizers for Continuous Probability
Distribution}\label{sec:opt_env_quant}
This section covers the system assumptions, cost formulation, and quantizer level
updates rules for the two nearly optimal quantization algorithms proposed, viz.
Approximate Lloyd Max (ALM) and the Approximate Envelope Quantizer (AEQ). The
cost function for ALM is based on the mean square error, while the cost function
for AEQ is the envelope constraint mean square error. The two cost functions are
iteratively minimized in a local neighborhood of the existing quantization
levels, so as to obtain a vectorized update policy (see
Sec.~\ref{subsec:ALM_opt_lev}, Sec.~\ref{subsec:AEQ_opt_lev}).
Sec.~\ref{subsec:unif_quant} provides insights into the proposed approximate
quantization scheme using the example of uniformly distributed sources. We
characterize the \emph{closeness} of the proposed approximate schemes to their respective
true (non-approximated) counterparts in an asymptotic sense, later in this
section (see
Sec.~\ref{subsec:asymp_opt}).

\subsection{Assumptions on the source distribution}
Let the data observations be generated from a known continuous probability
distribution $f_X(x)$. Without loss of generality, we assume that $f_X(x)$ is
supported on a unit interval, $[0,1]$ \footnote{For all practical data
distributions, the support is an interval $[a,b];\;\; a,b \in \mathbb{R}$.
Using scale and shift operations it can be transferred to $[0,1]$.}. The
following \emph{smoothness} criterion is assumed over the slope of the distribution
function:
\begin{align}
  \left|\frac{d}{d x} f_X(x)\right| < m \quad \text{ where } m \in (0,
\infty).\label{eq:slope_condition} 
\end{align}
The above condition states that the density function $f_X(x)$ has a slope
bounded by $m$. A scalar quantizer is defined by a map $Q: [0,1] \rightarrow
\{q_1,q_2,\cdots,q_K\}$. The image set of the map, $\{q_1,\cdots, q_K\}$
represent the discrete quantization levels. For elucidating our quantization
algorithm, we introduce two reference levels, $q_0:=0$ and $q_{K'}:=1$, which are
fixed at the endpoints of the interval $[0,1]$. In addition, we will assume the
following ordering, $q_0 \leq q_1 \leq \cdots \leq q_K \leq q_{K'}$ on the
quantization levels. The parameter $K$ above, is a fixed positive integer
denoting the number of quantization levels. In this paper the condition $K \gg
1$ is assumed. The subsequent discussion will elaborate on the cost functions and
their minimization techniques, that lead us to the two algorithms proposed in this
paper - ALM and AEQ.

\subsection{ALM cost minimization and level updates} \label{subsec:ALM_opt_lev}
The classical Lloyd-Max quantizer minimizes the mean square error (MSE) cost
function, by alternating the updates between the quantization level set $\{q_i; 1 \leq i
\leq K\}$ and the boundary level set $\{b_{j}; 1 \leq j \leq K+1\}$. The
above minimization can be equivalently performed by minimizing a collection of
local cost functions in the nearest left and right neighborhood of each
quantization level. This alternative approach is used to develop the ALM
algorithm.  The following cost function decomposition illustrates our approach:
\begin{align}
  \mathcal{R}(Q) :&= \eE\left[Q(X)-X\right]^2 \nonumber \\
  & = \int_0^1 \left(Q(x) -x\right)^2 f_X(x) \mbox{d}x \nonumber \\
  & = \sum_{k=1}^{K} \int_{b_{k}}^{b_{k+1}} \left(q_k-x\right)^2 f_X(x)\mbox{d}x
  \label{eq:ALM_cost}
\end{align}

In the above equation, the boundary set $\{b_j; 1 \leq j \leq
K+1\}$ defined as,
\begin{align*}
  b_{j+1} &= \frac{q_j + q_{j+1}}{2} \quad j =1,2,\cdots,K-1,\\
  b_1 &:= q_0 \quad \text{ and } \quad b_{K+1} := q_{K'}.
\end{align*}

We perform minimization of \eqref{eq:ALM_cost} by taking partial derivatives
with respect to each level in the set $\{q_i; 1 \leq i \leq K\}$. Using Leibniz
rule of differentiation under the integral sign we get the following condition
for the optimal levels \cite{Protter:2012}:
\begin{align}
  0  =  2 \int_{b_k}^{b_{k+1}} (q_k-x) f_X(x) \mbox{d}x,
  \label{eq:ALM_opt_cond}
\end{align}
for $1 \leq k \leq K$. The equation above, in general, does not give rise to a closed
form expression for $q_k$. Thus, we apply a piecewise linear approximation of
the density function to determine the approximate solution of
\eqref{eq:ALM_opt_cond}. 

\subsubsection{Optimal levels for approximate density}
\label{subsec:ALM_approx_den}
We consider the first order approximation of the density function in between the
nearest neighbor quantization levels. That is, 
\begin{align}
  f_{\app}(x) = m_k x+ c_k, \quad \text{ for } x \in [q_{k-1},q_{k+1}]
  \label{eq:lin_app}
\end{align}
where $m_k$ and $c_k$ corresponds to the slope and the intercept of the
approximation. These parameters are determined using the end point conditions
$f_{\app}(q_{k-1}) = f_X(q_{k-1})$ and $f_{\app}(q_{k+1}) = f_X(q_{k+1})$. The
linear approximation described above helps us to obtain a computable expression
for the optimal $q_k$ in \eqref{eq:ALM_opt_cond}. On replacing the density
function $f_X(x)$ by its approximation $f_{\app}(x)$ in the optimality condition (in
\eqref{eq:ALM_opt_cond}), a cubic equation, $r(u) = r_0 + r_1 u + r_2 u^2 + r_3
u^3$ is obtained, which has a real root in the interval $[q_{k-1},q_{k+1}]$
(See Appendix~\ref{subsec:ALM_exist_roots} for the proof). For $2 \leq k \leq K-1$, the
equation becomes quadratic as $r_3 =0$. The coefficients $r_0,r_1, r_2$ and
$r_3$ are tabulated for the different quantization levels in
Table.~\ref{tab:ALM_solve_cubic} 
\begin{table*}
  \centering
  \caption{\label{tab:ALM_solve_cubic}Coefficients of the cubic polynomial
  equation $r(u)= r_0 +r_1 u+r_2 u^2 + r_3 u^3$, to determine optimal
  level updates of ALM}
  \bgroup \def\arraystretch{1.5} \begin{tabular}{l|c|c|c}
    \toprule
     Coeff. & $k=1$ & $2 \leq k \leq K-1$ & $k = K$ \\[1.2ex]
    \midrule
    $r_0$ & $\begin{aligned} \frac{m_1}{3}\left(q_0^3-\frac{q_2^3}{8}\right) +
    \frac{c_1}{2}\left(q_0^2-\frac{q_2^2}{4}\right)\end{aligned}$ &
    $\begin{aligned}-\frac{m_k}{24}\left(q_{k+1}^3-q_{k-1}^3\right)-\frac{c_k}{8}\left(q_{k+1}^2
    -q_{k-1}^2\right)\end{aligned}$ & $\begin{aligned}
      \frac{m_K}{3}\left(\frac{q_{K-1}^3}{8}-q_{K'}^3\right) +
    \frac{c_K}{2}\left(\frac{q_{K-1}^2}{4}-q_{K'}^2\right) \end{aligned}$  \\[3ex] 
    $r_1$ & $\begin{aligned}-\frac{m_1}{2}q_0^2 + \frac{c_1}{4}q_2 -c_1q_0\end{aligned}$ &
    $\begin{aligned}\frac{c_k}{4}\left(q_{k+1}
      -q_{k-1}\right)\end{aligned}$ & $\begin{aligned}\frac{m_K}{2}q_{K'}^2 -
    \frac{c_K}{4}q_{K-1} +c_Kq_{K'}\end{aligned}$ \\[3ex]
    $r_2$ & $\begin{aligned}\frac{1}{8} m_1 q_2 + \frac{3}{8} c_1\end{aligned}$ &
    $\begin{aligned}\frac{m_k}{8}\left(q_{k+1}-q_{k-1}\right)\end{aligned}$ &
      $\begin{aligned}-\frac{1}{8} m_K q_{K-1} - \frac{3}{8} c_K\end{aligned}$  \\[3ex]
      $r_3$ & $\frac{1}{12}m_1$ & 0 & $-\frac{1}{12}m_K$ \\
      \bottomrule
  \end{tabular} \egroup
\end{table*}
\subsubsection{Insights into Approximate Lloyd-Max (ALM) quantization of uniformly
distributed sources} \label{subsec:unif_quant}
Let us consider a continuous source having a uniform distribution in the interval
$[0,1]$.  The optimal mean square error (MSE) quantizer for uniform distribution is
trivially obtained when the levels are fixed at equispaced locations on
the unit interval. However, we draw useful insights on the ALM algorithm when the
levels are initialized to random points. Starting from an initialized
quantization vector $q^{\;(0)}= \left[q_1^{\;(0)},\cdots,q_K^{\;(0)}\right]$, the ALM
algorithm minimizes the local cost function in the neighborhood interval,
$[q_{k-1},q_{k+1}]$ of each level $q_k$.  This results in a level update that
consistently reduces the overall MSE. In this specific example of sources with
uniform distribution, it is seen that the piecewise linear approximation
exactly represents the true distribution. The ALM cost minimization hence
follows the updates given as (see
Section.\ref{subsec:ALM_opt_lev},\ref{subsec:ALM_approx_den}),
\begin{align} 
  b_{i+1} &= \frac{q_{i}+q_{i+1}}{2} \quad \text{ for } i=1,2,\cdots,
  K-1 \nonumber \\ q_i &= \frac{b_{i-1} + b_i}{2} \quad \text{ for } i =
  1,2,\cdots,K.  \label{eq:LM_update}
\end{align}
It is noted that the same updates are obtained for the Lloyd-Max algorithm.
However, for general continuous distributions, the Lloyd-Max algorithm incur
additional computational expense due the evaluation of an integral for the
centroid update. Also, showing the global optimality of the Lloyd-Max
algorithm is analytically cumbersome for a general class of probability
distributions.  It will be later shown that the proposed algorithm is efficient
than the conventional Lloyd-Max update, since ALM gives a vectorized
rule using a series of matrix products. The motivation for our vectorized method
is derived from the \textit{even odd algorithm} due to Maheshwari and Kumar~\cite{Maheshwari:2015}. The original method although proposed for the data
driven quantizer design, extends to the model driven case that is considered in
this paper. The even odd algorithm updates the quantization levels in two steps.
In the first step, we modify the values of the odd set of levels,
$\{q_1,q_3,\cdots, q_{2l+1}\}$ while keeping the even indexed levels, viz.
$\{q_2,q_4,\cdots,q_{2m}\}$, fixed. This is followed by updating the even set,
fixing the odd set. This method speeds up the computation in each iteration of
the algorithm, as there are two sets of parallel update. Another advantage of
the even odd algorithm is its analytical simplicity. In each level
modification step, only the local neighbor points are considered. The above two
properties of the even odd algorithm, helps us to represent the overall vector
update  as a matrix transform. For example, consider the case where $K=3$.  Let
$\vec{q}^{\;(0)}$ represent the (randomly) initialized quantization levels.
Then the modified level after first iteration,
\begin{align}
  & \qquad\qquad\qquad\qquad \vec{q}^{\;(1)} = P_2 P_1 \vec{q}^{\;(0)},
  \label{eq:unif_mat_update} \\
  &\text{ where } \nonumber \\ 
  & P_1 = \begin{pmatrix}1 & 0 & 0 & 0 & 0 \\ \frac23 & 0 & \frac13 & 0 & 0
	\\ 0 & 0 & 1 & 0 & 0 \\ 0 & 0 & \frac13 & 0 & \frac23 \\ 0 & 0 & 0 & 0 &
	1\end{pmatrix}, \quad  
	P_2 = \begin{pmatrix}1 & 0 & 0 & 0 & 0 \\ 0 & 1 & 0 & 0 & 0 \\ 0 &
	\frac12 & 0 & \frac12 & 0	\\ 0 & 0 & 0 & 1 & 0 \\ 0 & 0 & 0 & 0 &
	1\end{pmatrix} . \nonumber 
\end{align}
The product matrix $P_2 P_1$ has the structure,
\begin{align}
  P = P_2P_1 = \begin{pmatrix} 1 & 0 & 0 & 0 & 0 \\ \frac23 & 0 & \frac13 & 0 &
    0 \\[1mm]
    \frac13 & 0 & \frac13 & 0 & \frac13 \\[1mm] 0 & 0 & \frac13 & 0 & \frac23 \\ 
  0 & 0 & 0 & 0 & 1\end{pmatrix}.
\end{align}
The matrices $P_1$ and $P_2$ are observed to be row stochastic with non-negative
entries. Hence, we can use the Perron-Frobenius theory
\cite{PFTheory:2005,Gallager:2012} to determine the fixed point of the product
matrix $P_2P_1$. The fixed point determined from $P_2P_1$, corresponds the
optimal solution which conforms with the optimal levels obtained using the
Lloyd-Max algorithm. The above result is realized on repeated application of the
update rule \eqref{eq:unif_mat_update}. The quantization level at the $n$-th
stage is given by update, $\vec{q}^{\;(n)} = (P_2P_1)^{n} \vec{q}^{\;(0)}$.  As $n
\rightarrow \infty$, $(P_2P_1)^n$ converges to a rank 2 matrix with two
non-zero columns. It will be later shown that these non-zero column vectors
corresponds to the fixed points of $P_2P_1$. A unique optima is obtained upon
imposing an ordering on the quantization levels. We show the following
properties for the matrix of interest, $P$.

\begin{enumerate} 
  \item $P$ is row stochastic.
  \item Eigenvalues of $P$ satisfy $|\lambda| \leq 1$.
  \item $\lambda =1$ is an eigenvalue and $\mathbf{1} = [1,1,\cdots,1]^T$ is a
    corresponding eigenvector.
  \item All eigenvectors of $P$ are either symmetric or antisymmetric.
  \item The geometric multiplicity of $\lambda = 1$ is 2; ie there are 2
    eigenvectors corresponding to the eigenvalue $1$.
  \item If $\vec{v}_1$ is an eigenvector of $\lambda = 1$, then
    $\vec{v}_2 = \mathbf{1}-\vec{v}_1$ is an independent eigenvector of $\lambda
    = 1$.
\end{enumerate}
The proof of the above results are discussed in
Appendix~\ref{appdx:unif_mat_prop}. Using these properties we can show that there
exists a fixed point such that $\vec{q}_{\opt} = \lim_{n \rightarrow \infty}
\vec{q}^{\;(n)}$ and $P\vec{q}_{\opt} = \vec{q}_{\opt}$. The proposed method is
noted for its exponential rate of convergence. For the example above, the decay rate is
$O\left(\frac{1}{3^n}\right)$, as the second largest eigenvalue of $P$ is $\frac13$. 

For an envelope constrained uniform quantizer, the update rules are similar to
the MSE quantizer, except for a level shift. In the following section we will
discuss the envelope quantization algorithm for continuous probability density,
based on insights derived from the quantization of uniform distributions.

\subsection{AEQ cost minimization and level updates}\label{subsec:AEQ_opt_lev}
Let the mean square error cost with the envelope constraint imposed be denoted
as $\mathcal{R}(Q)$. The following splitting of terms is possible on the terms
of $\mathcal{R}(Q)$:
\begin{align}
  \mathcal{R}(Q) &:= \eE\left[(Q(X)-X)^2\right] \quad \text{where } Q(X) \geq X,
  \nonumber\\
  & = \int_0^{1} \left(Q(x)-x\right)^2 f_X(x) \mbox{d}x \quad \text{where } Q(x)
  \geq x,
  \nonumber \\
  & = \sum_{k=1}^{K} \int_{q_{k-1}}^{q_k} \left(q_k-x\right)^2 f_X(x) \mbox{d}x
  . \label{eq:cost_env}
\end{align}
The simplification in the cost function above is performed by substituting $Q(x)
= q_k$ for $q_{k-1} \leq x < q_k$. It is observed that the total cost can be
minimized with respect to the each quantization level $q_k; \;k =1,\cdots, K$.
The minima corresponds to equating the partial derivative to zero.  That is,
\begin{align}
  0 = & \frac{\partial \mathcal{R}(Q)}{\partial q_k} \nonumber \\
  = & \frac{\partial}{\partial q_k}  \int_{q_{k-1}}^{q_k} (q_k-x)^2 f_X(x)
  \mbox{d}x \nonumber \\ & + \frac{\partial}{\partial q_k}
  \int_{q_{k}}^{q_{k+1}} (q_{k+1}-x)^2 f_X(x) \mbox{d}x \nonumber\\ 
  = &  \int_{q_{k-1}}^{q_k} 2(q_k-x)f_X(x) \mbox{d}x - (q_{k+1}-q_k)^2f_X(q_k)
  \label{eq:quant_level_cond}
\end{align}
In the above equation, it is noted that corresponding to $q_k$, the nearest
neighbor levels $q_{k-1}$ and $q_{k+1}$ are fixed. This implies that the
quantizer level updates can be performed simultaneously for all even (or odd)
indices, while fixing the odd (or even) indices. Since the
modified quantization levels can be determined by two separate parallel updates
of even and odd sets, we term this procedure as \textit{Alternating Between
Evens and Odds (ABEO)}.  This update rule will considerably speed up the
proposed quantization approach.  It is observed that, in general
\eqref{eq:quant_level_cond} does not ensure a closed-form solution of $q_k$.
Hence we provide a linear approximation based algorithm for envelope
quantization. This is considered in the next section.

\subsubsection{Linear Approximation Based
Algorithm}\label{subsec:AEQ_approx_den}
The linear approximation described for ALM (in \eqref{eq:lin_app}) helps us to obtain a
closed-form expression for the optimal $q_k$ for the AEQ levels. We rewrite the
sufficient condition for optimality using the approximate density function $f_{\app}(x)$.
\begin{align}
  0 = \int_{q_{k-1}}^{q_k} 2(q_k-x)f_{\app}(x) \mbox{d}x -
  (q_{k+1}-q_k)^2f_{\app}(q_k) \label{eq:linapp_cond}
\end{align}
We substitute \eqref{eq:lin_app} in the above equation, to obtain a third order
polynomial equation, $p(u) = p_0 + p_1 u + p_2 u^2 + p_3 u^3$. The coefficients
$p_j;\;j =0,1,2,3$ depends on the nearest neighbor levels, $q_{k-1}$ and
$q_{k+1}$. We list these coefficients in Table~\ref{tab:AEQ_solve_cubic}.
\begin{table}
  \centering
  \caption{\label{tab:AEQ_solve_cubic}Coefficients of the cubic polynomial
  equation $p(u)= p_0 +p_1 u+p_2 u^2 + p_3 u^3$, to determine optimal
  level updates of AEQ}
  \bgroup \def\arraystretch{1.5} 
  \begin{tabular}{l|c}
    \toprule
    Coeff. &  $1 \leq k \leq K$ \\[1.2ex]
    \midrule
    $p_0$ &  $\begin{aligned}\frac23 m_k q_{k-1}^3 +
    c_k\left(q_{k-1}^2-q_{k+1}^2\right)\end{aligned}$ \\
  $p_1$ & $\begin{aligned}-2c_k(q_{k-1}-q_{k+1}) - m_k(q_{k+1}^2+q_{k-1}^2)
  \end{aligned}$ \\
  $p_2$ & $\begin{aligned} 2 m_k q_{k+1} \end{aligned}$ \\
  $p_3$ & $\begin{aligned}-\frac23 m_k \end{aligned}$ \\
    \bottomrule
  \end{tabular} \egroup
\end{table}

The roots of the cubic equation $p(u)=0$ in the interval $[q_{k-1},q_{k+1}]$,
corresponds to the optimum level update of $q_k$. The existence of
atleast one real root in $[q_{k-1},q_{k+1}]$ is shown in the
Appendix~\ref{subsec:AEQ_exist_roots}. The proposed linear approximation based
quantization scheme is described in a step-wise manner in
Algorithm.~\ref{algo:scalar_quant}.  
\begin{algorithm}\label{algo:scalar_quant}
  \caption{Scalar Envelope Quantizer Algorithm}
  \SetKwInOut{Input}{Input}
  \SetKwInOut{Output}{Output}
  \SetKwInOut{Init}{Initialization}

    \Input{Input distribution $f_X(x)$, $K =$ \# of levels, MaxIter, Threshold}
  \Output{List of quantization levels, $\vec{q}$}
    \Init{$\vec{q}^{\;(0)} = \left[0,\frac{1}{K}, \frac{2}{K}, \cdots, 1\right]$, stop
  condition = False, $i = 0$, dist = 0}

  \While{!stop condition}{
    $\mathcal{Q}_{\odd} = \{q_1, q_3, \cdots, q_{2m+1}\}$ \\
    $\mathcal{Q}_{\even} = \{q_2, q_4, \cdots, q_{2l}\}$ \\
    $\qquad$    \% where $\max\{2l,2m+1\}=K-1$\\ 
    \For{(In Parallel) $q_k$ in $Q_{\odd}$}{
      $|$ \\
      Set linear approximation parameters :
      \begin{align}
	& \mbox{ Slope: } m_k = \frac{f_X(q_{k+1})-f_X(q_{k-1})}{q_{k+1}-q_{k-1}},
	\nonumber\\
	& \mbox{ Intercept: } c_k = f_X(q_{k+1})-m_k q_{k+1}, \text{ and } \nonumber\\  
	& p(u) = p_0 + p_1 u + p_2 u^2 + p_3 u^3  \text{ (see Table
	\ref{tab:AEQ_solve_cubic})} \nonumber
      \end{align}
      $q_{k}^{(i+1)} \leftarrow \{v \in [q_{k-1},q_{k+1}] : p(v)=0 \}$
      $\color{white}{1}\;\;\;\;\; \qquad \color{black}{\text{Note:}
      (\mbox{Im}(r)=0\})}$ \\
    $|$
      }
    \For{(In Parallel) $q_k$ in $\mathcal{Q}_{\even}$}{
      $|$ \\
      Update $q_k$ with steps in $\mathcal{Q}_{\odd}$ loop above \\
      $|$
      }
    $\mbox{dist} \leftarrow \mathcal{R}(\vec{q}^{\;(i+1)})$;   $\quad i \leftarrow
    i+1$

    \If{(dist $<$ Threshold) or (iter $>$ MaxIter)}{
      stop condition = True
    }
  }
\end{algorithm}
For the ALM scheme, the steps of Algorithm.~\ref{algo:scalar_quant} is valid, when the level
modification step in \textit{line~8} is performed using the cubic polynomial
$r(u)$ in Table.~\ref{tab:ALM_solve_cubic} instead of $p(u)$. The ALM and AEQ levels obtained
using the piecewise linear approximation will be shown to be arbitrarily close
to their respective true levels as $K \rightarrow \infty$.
\subsubsection{Asymptotic Optimality of Piecewise Linear
Approximation}\label{subsec:asymp_opt}
Let $\vec{q}^{\;\ast}(.)$ be the optimal quantizer with respect to the true density
function and $\vec{q}_{\mbox{\scriptsize A}}^{\;\ast}(.)$ be the quantizer obtained by the linear approximation
of the density $f_X(x)$. The asymptotic convergence (as quantization levels $K
\rightarrow \infty$) of the linear approximation scheme can be established by
using the Taylor series expansion. At $x= q_k$, the Taylor approximation around
the interval $x \in [q_k-\delta/2,q_k+\delta/2]$ and $\delta >0$, is given by
\begin{align} 
  f_X(x) & = f_X(q_k) + f_X'(q_k)(x-q_k) + O((x-q_k)^2) \nonumber \\ 
  & = m_k x + c_k + O\left((x-q_k)^2\right)  \nonumber \\ 
  & = f_{\app}(x) + O\left((x-q_k)^2\right).  
  \label{eq:Taylor_app} 
\end{align}
For the simplifying the notations in the analysis, we restrict our attention to
the value of $k=2$. The two neighboring levels of interest are $q_1$ and $q_3$.
Let $q_2^{\ast}$ and $q_{2\mbox{\scriptsize A}}^{\ast}$ denote the optimal level
updates of the true density and the approximated density respectively (see
\eqref{eq:quant_level_cond} and \eqref{eq:linapp_cond}). Then, the Taylor expansion at
$x = q_2^{\ast}$ is,
\begin{align} f_X(q_2^{\ast}) = f_{\app}(q_2^{\ast}) +
  O(\varepsilon_K), \label{eq:Taylor_app2}
\end{align}
where $\varepsilon_K = \max_{1 \leq k \leq K-1} |q_{k+1} - q_{k-1}|^2$. Using
the above fact, $|q_{2}^* - x|^2 \leq |q_{3}-q_1|^2 \leq \varepsilon_K$ for all $x
\in (q_1,q_3)$. For the specific example of a uniformly distributed source,
$\varepsilon_K = \frac{1}{K^2}$. The asymptotic optimality of the ALM and AEQ
schemes, as $K \rightarrow \infty$ is summarized in the following result.
\begin{theorem}[\textit{Asymptotic optimality of ALM and AEQ}]
  The approximate solution of the quantization level update (see
  Table.~\ref{tab:ALM_solve_cubic} for ALM and Table.~\ref{tab:AEQ_solve_cubic}
  for AEQ), $q_{2\mbox{\scriptsize A}}^{\ast}$ converges to the true solution,
  $q_2^{\ast}$ as the number of levels $K \rightarrow \infty$. That is, there
  exists a $K \geq K_0$ such that $|q_{2\mbox{\scriptsize A}}^{\ast}-q_2^{\ast}|
  \leq \varepsilon$ for all $\varepsilon > 0$.
\end{theorem}
\begin{proof} 
  The proofs for ALM and AEQ schemes are dealt separately below.

  \textit{ALM Optimality} : Consider the level update expression in
  \eqref{eq:ALM_opt_cond} evaluated at the true solution $q_2^{\ast}$. Using the
  fact that $b_{2}:=\frac{q_1+q_2}{2}$, the expressions corresponding to the
  true density, $f_X(x)$ and its approximation, $f_{\app}(x)$ are given by,
  \begin{align}
    0 = 2 \int_{b_2}^{b_3} (q_2^{\;\ast}-x)f_X(x) \mbox{d}x &  \label{eq:AEQ_asymp_cond1} \\
    D = 2 \int_{b_2}^{b_3} (q_2^{\;\ast}-x)f_{\app}(x) \mbox{d}x & \label{eq:AEQ_asymp_cond2}
  \end{align}
  On subtracting the \eqref{eq:AEQ_asymp_cond1} from \eqref{eq:AEQ_asymp_cond2}
  and using the Taylor approximation in \eqref{eq:Taylor_app2}, we show,
  \begin{align}
    D &= \frac14 (q_{3}-q_1) \left(q_2^{\;\ast}-\frac{q_3+q_1}{2}\right)
    O(\varepsilon_K) \nonumber \\
    &= O(\varepsilon_K^2)
  \end{align}
  We use the fact that the local MSE cost function about $q_2$ is continuous
  and has a positive second derivative. By the continuity of the cost function,
  we infer that $q_{2\mbox{\scriptsize A}}^{\;\ast} \rightarrow q_2^{\;\ast}$
  as $K \rightarrow \infty$.
  \textit{AEQ Optimality} : We note that the true solution $q_2^{\ast}$ satisfies the optimality equation
  in  \eqref{eq:quant_level_cond}. On applying $q_2 = q_2^{\ast}$ in
  \eqref{eq:linapp_cond}, the rule is satisfied with an offset value $D$. That
  is, 
  \begin{align}
    \int_{q_{1}}^{q_2^{\ast}} 2(q_2^{\ast}-x) f_X(x) \mbox{d}x - (q_{3}-q_{2}^{\ast})^2
    f_X(q_2^{\ast}) &= 0 \label{eq:asymp_opt_1}\\
    \int_{q_{1}}^{q_2^{\ast}}2(q_2^{\ast}-x) f_{\app}(x) \mbox{d}x - (q_{3}-q_{2}^{\ast})^2
    f_{\app}(q_2^{\ast}) &= D \label{eq:asymp_opt_2}
  \end{align}
  Subtracting \eqref{eq:asymp_opt_1} from \eqref{eq:asymp_opt_2}, we get
  \begin{align}
    D & =   \int_{q_{1}}^{q_2^{\ast}}\left[2(q_2^{\ast}-x)\right]
    O\left(\varepsilon_K\right)
    \mbox{d}x - \left[(q_{3}-q_{2}^{\ast})^2\right] O\left(\varepsilon_K\right)
    \nonumber \\
    & = 2\left(q_3-q_1\right)\left(q_2^{\ast} - \frac{q_1+q_3}{2}\right)
    O\left(\varepsilon_K\right) \nonumber \\
    & = O\left(\varepsilon_K^2\right)
  \end{align}
  The above result shows that the offset $D$ as a function of the optimal
  solution $q_2^{\ast}$ eventually converges to zero. Since the approximate
  solution, $q_{2\mbox{\scriptsize A}}^{\ast}$ is the root of the left hand side
  of \eqref{eq:asymp_opt_2}, we argue that $q_{2\mbox{\scriptsize A}}^{\ast}$
  approaches $q_{2}^{\ast}$ arbitrarily close, as $K \rightarrow \infty$. The
  fact is true since the AEQ optimality condition \eqref{eq:linapp_cond} is
  continuous in $q_2$ and has a positive derivative at $q_{2\mbox{\scriptsize
  A}}^{\ast}$ (see Appendix~\ref{appdx:AEQ_pos_derivative} for proof). The above properties
  ensure that, $q_{2\mbox{\scriptsize A}}^{\ast}  \rightarrow q_2^{\ast}$ as $D
  \rightarrow 0$.
\end{proof}
\begin{remark}[\textit{An alternate bound on ALM and AEQ near optimality result}]
  The absolute difference of the nearly optimal and the true optimal values of
  $q_2$ are bounded by the maximum length of the interval
  $[q_{k-1},q_{k+1}];\;k=1,2,\cdots,K-1$. That is, \[|q_{2\mbox{\scriptsize
  A}}^{\ast} - q_2^{\ast}| \leq |q_3-q_1| \leq \sqrt{\varepsilon_K}.\] This
  gives a loose bound on the near optimality, which follows directly from the
  decreasing interval length with $k$.
\end{remark}
The above result holds true when $k=2$ is replaced by any $1\leq k \leq K-1$.
Hence, we see that the approximated quantization vector
$\vec{q}_{\mbox{\scriptsize A}}^{\;\ast}$
converges to the true quantization vector $\vec{q}^{\;\ast}$ under the
$\ell_\infty$ norm. That is, \[\|\vec{q}_{\mbox{\scriptsize A}}^{\;\ast}-
\vec{q}^{\;\ast}\|_{\infty}:= \max_{1 \leq k \leq K} |q_{k\mbox{\scriptsize
A}}^{\ast}-q_{k}^{\ast}| \rightarrow 0\]

In simulations (shown in Fig.~\ref{fig:variation_plots}(g)-(i)) ,
it is observed that the quantization levels obtained from the approximation
schemes are close to the true optima, computed using the original density
function. The ALM and AEQ schemes proposed here achieves a nearly
optimal solution, with a reduced computational burden.

From this point, we treat the analysis of ALM and AEQ algorithms in a common
framework. We abstract out the solution of the cost minimization procedure, that
is the roots of the polynomials in Table~\ref{tab:ALM_solve_cubic}
and Table~\ref{tab:AEQ_solve_cubic}, and express the resulting levels shifts as
linear transformations.
\subsection{Level shifts as linear updates}

The optimal solution for the iterative update of level $q_k$ is given by the
roots of \eqref{eq:ALM_opt_cond} or \eqref{eq:linapp_cond} in the interval
$[q_{k-1},q_{k+1}]$. The solution at $i$-th iteration can be expressed as a
convex combination,
\begin{align}
  q_{k}^{(i+1)} = \theta_k^{(i)} q_{k-1}^{(i)} + (1-\theta_k^{(i}) q_{k+1}^{(i)}
  \quad \text{ where } \theta_{k}^{(i)} \in [0,1]
  \label{eq:cvx_comb}.
\end{align}
The above update equation will aid in the convergence analysis of the proposed
algorithm. In vector notation the ABEO update rule can be expressed
as,
\begin{align}
  \vec{q}^{\;(i+1)} = P_{\odd}^{(i)} P_{\even}^{(i)} \vec{q}^{\;(i)} \quad \text{ where }
  i=0,1,\ldots. \label{eq:vec_update}
\end{align}
In the above equation $P_{\even}^{(i)}$ and $P_{\odd}^{(i)}$ are square matrices
having dimension $K'+1$. Note that for the ALM scheme $K'=K+1$ and for AEQ
$K'=K$. These square matrices determines the optimal level updates obtained
using \eqref{eq:cvx_comb}. For instance, for $K' = 4$,
\begin{align}
  P_{\odd}^{(i)} & = \begin{pmatrix}1 & 0 & 0 & 0 & 0 \\ \theta_1^{(i)} & 0 & 1-
    \theta_1^{(i)} & 0 & 0 \\ 0 & 0 & 1 & 0 & 0 \\ 0 & 0 & \theta_3^{(i)} & 0 &
  1-\theta_3^{(i)}\\ 0 & 0 & 0 & 0 & 1\end{pmatrix}, \nonumber  \\
  P_{\even}^{(i)} & = \begin{pmatrix}1 & 0 & 0 & 0 & 0 \\ 0 & 1 & 0 & 0 & 0 \\ 0
    & \theta_2^{(i)} & 0 & 1-\theta_2^{(i)} & 0	\\ 0 & 0 & 0 & 1 & 0 \\ 0 & 0 &
  0 & 0 & 1\end{pmatrix}.  \label{eq:Podd_mat} 
\end{align}
We note that the two matrices, $P_{\odd}$ and $P_{\even}$ are row stochastic. A
(row) symmetry on the location of the zeros is also observed. The matrix operators
preserve the values of reference levels, $q_0$ and $q_{K}$ in every iteration.
This is attributed to the first and the last rows of the \eqref{eq:Podd_mat} The
vector update of the quantizer explained in
\eqref{eq:vec_update}-\eqref{eq:Podd_mat}, has got the required structure to
apply convergence using the Perron Frobenius theory
\cite{PFTheory:2005,Gallager:2012}.
\section{Convergence Analysis of Near Optimal Quantizers}\label{sec:conv_env_quant}
This section describes the analysis for convergence of the linear approximation
based method in Algorithm \ref{algo:scalar_quant}. Using the fact that,
the product of two row stochastic matrices is row stochastic, we show that
$P^{(i)} := P_{\even}^{(i)} P_{\odd}^{(i)}$ has every row that sums to unity. For
the $K=4$ case, the above product matrix has the following structure:
\begin{align}
  P^{(i)} = \begin{pmatrix} 1 & 0 & 0 & 0 & 0 \\
    \theta_1^{(i)} & 0 & \bar{\theta}_1^{(i)} & 0 & 0 \\
    \theta_2^{(i)} \theta_1^{(i)} & 0 & \bar{\theta}_1^{(i)} \theta_2^{(i)} +
    \bar{\theta}_2^{(i)} \theta_3^{(i)} &
    0 & \bar{\theta}_2^{(i)} \bar{\theta}_3^{(i)} \\
    0 & 0 & \theta_3^{(i)} & 0 & \bar{\theta}_3^{(i)} \\
    0 & 0 & 0 & 0 & 1
  \end{pmatrix}, \label{eq:P}
\end{align}
where $\bar \theta_k^{(i)} = 1 -\theta_k^{(i)}$ is used for
concise notation. Other properties from the individual matrices, such as (row) symmetry
on zero locations, are carried forward to the $P^{(i)}$ matrix. The first and
last rows of the matrix are independent of the scale parameters
$\theta_j;\;\;j=1,\cdots,K-1$. An important observation is regarding the zero
vectors, that appear alternatively in the columns of the above matrix. This
occurs due to the fact that the linear updates $P_{\odd}^{(i)}$ and
$P_{\even}^{(i)}$,  acts only on the alternate entries (nearest neighbors) of
the quantization vector $\vec{q}^{\;(i)}$. An important observation on $P^{(i)}$
is that $0 \geq [P^{(i)}]_{l,m} \leq 1$ for all entries $(l,m)$. Further, we
show that the coefficient of the linear combination $\theta_{j}^{(i)}$ for $j =
1, 2, \cdots, K-1$ are bounded away from the extremes $0$ and $1$. This is shown
in the Fact below.
\begin{prop}[$\theta_j$'s are bounded away from extremes]\label{lem:bdd_away}
  The coefficients $\theta_{j}^{(i)}$ for $j = 1, 2, \cdots, K-1$ in \eqref{eq:P}
  satisfy the criteria $0 <\theta_{j}^{(i)} < 1$ for all iteration count $i$.
\end{prop}
\begin{proof}
  From the smoothness assumption \eqref{eq:slope_condition}, we deduce the fact
  that the linear approximation slope is bounded, that is, $|m_k| \leq B$.
  First, we consider the case when $m_k =0$. From the ALM and AEQ optimality
  conditions (see \eqref{eq:ALM_opt_cond} and \eqref{eq:linapp_cond}) we observe
  that $q_k = \frac12 q_{k-1} + \frac12 q_{k+1}$. In other words, a flat density
  approximation results in $\theta_k = \frac12$.   
  
  In the second case, when $m_k \neq 0$ and $|m_k| \leq B$, we show that
  $\theta_k \in (0,1)$. To establish the result, we first observe the fact that
  the ALM and AEQ solutions always lie in the interval $[q_{k-1},q_{k+1}]$ (see
  Appndix \ref{subsec:ALM_exist_roots}, \ref{subsec:AEQ_exist_roots}). Thus
  $\theta_k \in [0,1]$. We now consider the boundary cases corresponding to
  $\theta_k=0$ and $\theta_k=1$. These are equivalent to the solution $q_k =
  q_{k+1}$ and $q_k = q_{k-1}$ respectively. It is observed that under these
  extreme cases the ALM and AEQ optimality conditions (see  \eqref{eq:ALM_opt_cond} and
  \eqref{eq:linapp_cond}) are valid only when $q_{k+1} = q_{k-1}$; which
  corresponds to a trivial case. The resulting contradiction, hence shows that
  $\theta_k \in (0,1)$, for all bounded values of the slope $m_k$.
\end{proof}

In accordance with Algorithm~\ref{algo:scalar_quant}, the quantization levels
after $L$ iterations is $\vec{q}^{\;(L)} = \prod_{i=1}^{L} P^{(i)} \vec{q}^{\;(0)}$.
We show the following convergence property of the product $\prod_{i=1}^L
P^{(i)}$ below.
\begin{prop}[Convergence of columns of product matrix]\label{lem:prod_conv}
  The odd columns $\{\vec{c}_k :k=3,5,\cdots,2m+1\}$ of the sequence
  $\prod_{i=1}^L P^{(i)}$ converges to zero as $L \rightarrow \infty$.
\end{prop}
\begin{proof}
  We consider the first two terms of the product, viz. $P^{(1)}$ and $P^{(2)}$.
  The two matrices are row stochastic. Any element of the product
  $P^{(2)}P^{(1)}$ is an inner product between a row vector of $P^{(2)}$ and a
  column vector of $P^{(1)}$. Let $\vec{v}_r^{(2)}$  and $\vec{u}_s^{(1)}$ be
  the (column vector) representation of the $r$-th row of $P^{(2)}$ and the
  $s$-th column of $P^{(1)}$ respectively. Then, $w_{r,s} =
  [P^{(2)}P^{(1)}]_{r,s}$, the $(r,s)$ entry of the product $P^{(2)}P^{(1)}$, is
  the inner product between $\vec{v}_r^{(2)}$ and $\vec{u}_s^{(1)}$. That is, $w_{r,s} =
  \langle\vec{v}_r^{(2)},\vec{u}_s^{(1)}\rangle$. The following facts hold true
  for $\vec{v}_r^{(2)}$ and $\vec{u}_s^{(1)}$:
  \begin{enumerate}
    \item[(F1)] $\begin{aligned}\max_{1 \leq t \leq K+1}[\vec{u}_s^{(1)}]_t <
    1\end{aligned}$ for every $s \neq 1$ and $n \neq K+1$.
      (See Proposition.~\ref{lem:bdd_away}) 
    \item[(F2)] $\mathbf{1}^T\vec{v}_r^{(2)} = 1$ where $\mathbf{1} =
      [1,1,\cdots,1]^T$. This is true since $P^{(2)}$ is a row stochastic
      matrix.
  \end{enumerate}
  Since the inner product $w_{r,s}$ is a (non-zero) convex combination (from (F2)
  and due to Proposition.~\ref{lem:bdd_away}) of
  components of $\vec{u}_s^{(1)}$, from (F1) we assert that 
  \begin{align}
    0 \leq w_{r,s} < \max_{1 \leq t \leq K+1}  [\vec{u}_s^{(1)}]_t.
    \label{eq:contra_prod_mat}
  \end{align}
  In other words, all the elements of columns $\{c_k : k=3,5,\cdots,2l+1\}$ are
  strictly less than the maximum element in the corresponding columns of
  $P^{(1)}$. Alternatively, we can represent this as a contraction,
  $\begin{aligned}w_{r,s} = \alpha \max_{t}  [\vec{u}_s^{(1)}]_t\end{aligned}$,
    where $\alpha < 1$.
    Using the fact that product $P^{(2)} P^{(1)}$ is row stochastic, we can
    extend the same argument to the product of three matrices, that is
    $P^{(3)}[P^{(2)}P^{(1)}]$. We make use of an induction
  argument to show the property in the limiting case. Let $w_{r,s}^{(n)}$ be the
  $(r,s)$-th element of the product matrix $\prod_{i=1}^n P^{(i)}$. Then using
  \eqref{eq:contra_prod_mat}, we have the contraction of the sequence,
  \begin{align} 
    w_{r,s}^{(n)} = \left(\prod_{i=1}^n \alpha_i\right)  \max_{
      1\leq s \leq K+1} w_{r,s}^{(1)},
  \end{align}
  where $\alpha_i < 1$. From above we observe that, the
  sequence of inner product terms, $\{w_{r,s}^{(i)}: 0 \leq i < \infty\}$ is a
  monotonically decreasing sequence for every $1 \leq r \leq K+1$ and $s =
  3,5,\cdots,2l+1$. Using monotone convergence theorem \cite{Rudin:2006} on the
  above (bounded) sequence, we show that the columns $\vec{c}_3,\vec{c}_5,
  \cdots, \vec{c}_{2l+1}$ of the product sequence $\prod_{i=1}^{L}  P^{(i)}$
  eventually decreases to zero.
\end{proof}
\begin{theorem}[\textit{Convergence to global minima}]
  The iteration in Algorithm.~\ref{algo:scalar_quant} converges to a
  quantization vector,
  \begin{align}
    \vec{q}^{\;\ast} = P^{\ast} \vec{q}^{\;(0)}, 
  \end{align}
  where $P^{\ast} = \lim_{L \rightarrow \infty} \prod_{i=1}^{L} P^{(i)}$, and
  $\vec{q}^{\ast}$ is independent of the initialization $q^{(0)}$ (except for the
  reference levels).  
\end{theorem}
\begin{proof}
  We note that the limiting product matrix $P^{\ast} = \lim_{L \rightarrow \infty}
  \prod_{i=1}^{L}P^{(i)}$ converges to a matrix with columns $\vec{c}_2 =
  \vec{c}_3 = \cdots = \vec{c}_{K} = \mathbf{0}$ (See Proposition.~\ref{lem:prod_conv}).
  The first and last columns, viz. $\vec{c}_1$ and $\vec{c}_{K+1}$, are non-zero
  vectors, as the matrix transformation $P^{(i)}$, preserves the reference levels $q_0$ and
  $q_K$. The structure of $P^{\ast}$ is as follows:
  \begin{align} 
    P^{\ast} = \begin{pmatrix} \kern.6em\vline & \kern.6em\vline  & \cdots &
      \kern.6em\vline &  \kern.6em\vline \\ \vec{c}_1 & 0 & \cdots & 0 &
      \vec{c}_{K+1} \\ \kern.6em\vline & \kern.6em\vline  & \cdots &
      \kern.6em\vline &  \kern.6em\vline \end{pmatrix}
  \end{align}
  Recalling the fact that $P^{\ast}$ is a row stochastic matrix, the all ones
  vector $\mathbf{1}$, is an eigenvector of $P^{\ast}$ corresponding to
  eigenvalue $\lambda =1$. Due to this fact, we can show $\vec{c}_1+
  \vec{c}_{K+1} = \mathbf{1}$. In other words, the elements of $\vec{c}_{1}$ and
  $\vec{c}_{K+1}$ forms a convex combination pair. From the above fact, we also infer
  that, $\vec{c}_{K+1}$ is \textit{order reversed} with respect to vector $\vec{c}_1$.
  The above two column vectors are also linearly independent and correspond to
  the eigenvectors of $\lambda =1$. Using the Gauss elimination method, we can
  establish that the matrix $P^{\ast}$ has a rank of 2. Since there is a repeated
  eigenvalue $\lambda=1$ with geometric multiplicity 2, all the remaining
  eigenvalues of the limiting matrix are zero. Hence, $P^{\ast}$ has two fixed points
  $\vec{c}_{1}$ and $\vec{c}_{K+1}$.

  On imposing an order constraint on the quantization levels $0:=q_0 \leq q_1
  \leq \cdots \leq q_{K} :=1$, we can show that $\vec{c}_{K+1}$ corresponds to
  the global minimizer of the Algorithm.\ref{algo:scalar_quant}. The above
  iterative scheme has an exponential rate of convergence, as all the
  eigenvalues of the transform matrix $P^{(i)}$ satisfy the property that
  absolute value of eigenvalues atmost 1.

  The initialization $\vec{q}^{\;(0)}$ (assuming $q_0^{(0)}=0$ and $q_K^{(0)}=1$), has no
  effect on the fixed point of $P^{\ast}$, as all the intermediate columns of
  the limiting matrix are zero vectors.  However, the transient terms of the
  product $\prod_{i=1}^{L} P^{(i)}$, will depend on the initialized quantization
  levels for small values of $L$. 
\end{proof}

\begin{remark}[\textit{Uniqueness of solution}]
  The optimal quantization vector given by Algorithm.~\ref{algo:scalar_quant} is
  unique upto an ordering.
\end{remark}
\begin{remark}[\textit{Exponential rate of convergence}]
  The linear approximation based algorithm (see
  Algorithm.~\ref{algo:scalar_quant}) achieves exponential rate of
  convergence, that is, the $\ell_2$ gap between the linear approximation solution and
  the true solution drops at a rate $O(\lambda_{(2)}^n)$, where $\lambda_{(2)}$
  represents the second largest eigenvalue of the transform matrix, $P^{(i)}$.
\end{remark}
\section{Simulation Results and Discussion} \label{sec:sim_results}

We present the simulation results of the model based scalar quantization in this
section. The results corresponds to finite support distributions in the interval
$[0,1]$. All simulations were performed using Numpy module in Python 3.5 kernel
with normal computing hardware (Intel i7, 2.2 GHz processor, 8GB RAM). We study
the variations of the MSE with different simulation parameters for the
unconstrained and constrained cases (see Fig.~\ref{fig:variation_plots}). The
density functions considered in the simulation include the Beta distribution,
truncated normal and truncated exponential distributions.

The quantizer evolution plot for the envelope quantizer on the Beta(2,4) is
shown in Fig.~\ref{fig:evolution_env_2_4}. The simulation considers $K=8$
levels, with an equi-spaced initialization. The levels are seen to shift towards
the peak near $x=0.25$, in a asymmetric manner. The reference levels at $x=0$
and $x=1$ are unchanged during the iteration of algorithm. The levels
obtained by this linear approximation algorithm are observed to be close to
optimal levels without the approximation. In Fig.~\ref{fig:variation_plots}
(a), a comparison of MSE variation for the envelope quantizer and the
unconstrained quantizer is shown. The convergence of the iterative algorithm is
shown in Fig.~\ref{fig:variation_plots}(b) for different values of $K$. It is
noted that the convergence for each of the plots happens in approximately $K$
iterations. Fig.~\ref{fig:variation_plots}(c) describes the
MSE performance of the algorithm for different symmetric  Beta distributions. The
MSE decays at exponential rate with the number of quantization levels. The
advantage of using linear approximation scheme is shown in
Fig.~\ref{fig:variation_plots}(d), by comparing the simulation time required to
run the quantization algorithm. In this experiment we used the stopping criteria
for the algorithm to be the iteration until the computed MSE is within $1\%$ of
the optimal MSE. We observe that the ALM has an average of $3.4$x improvement
with respect to the simulation time (using Beta(4,2) distribution).
Fig.~\ref{fig:variation_plots}(e)-(f) shows the MSE variation of the envelope
quantizer for different truncated distributions and asymmetric Beta
distributions respectively. In Fig.~\ref{fig:variation_plots}(g)-(i), the
quantization levels for the different MSE cost metrics is compared. It is
observed that the linear approximation method closely tracks the optimal levels
obtained without using the approximation. The difference between the
unconstrained MSE and the envelope quantizer is explained by the quantization
levels shown in Fig.~\ref{fig:variation_plots}(g) and (h). For the Beta(2,2)
distribution, the unconstrained quantizer results in a symmetric set of
quantization levels, while the envelope quantizer has a right shifted and
asymmetric set of levels. For the asymmetric Beta(2,4) distribution, the
envelope quantizer allots more number of levels around the region where the
density $f_X(x)$ peaks (see Fig.~\ref{fig:variation_plots}(i)). 
\begin{figure}[!hbt]
  \centering
  \includegraphics[scale=0.55]{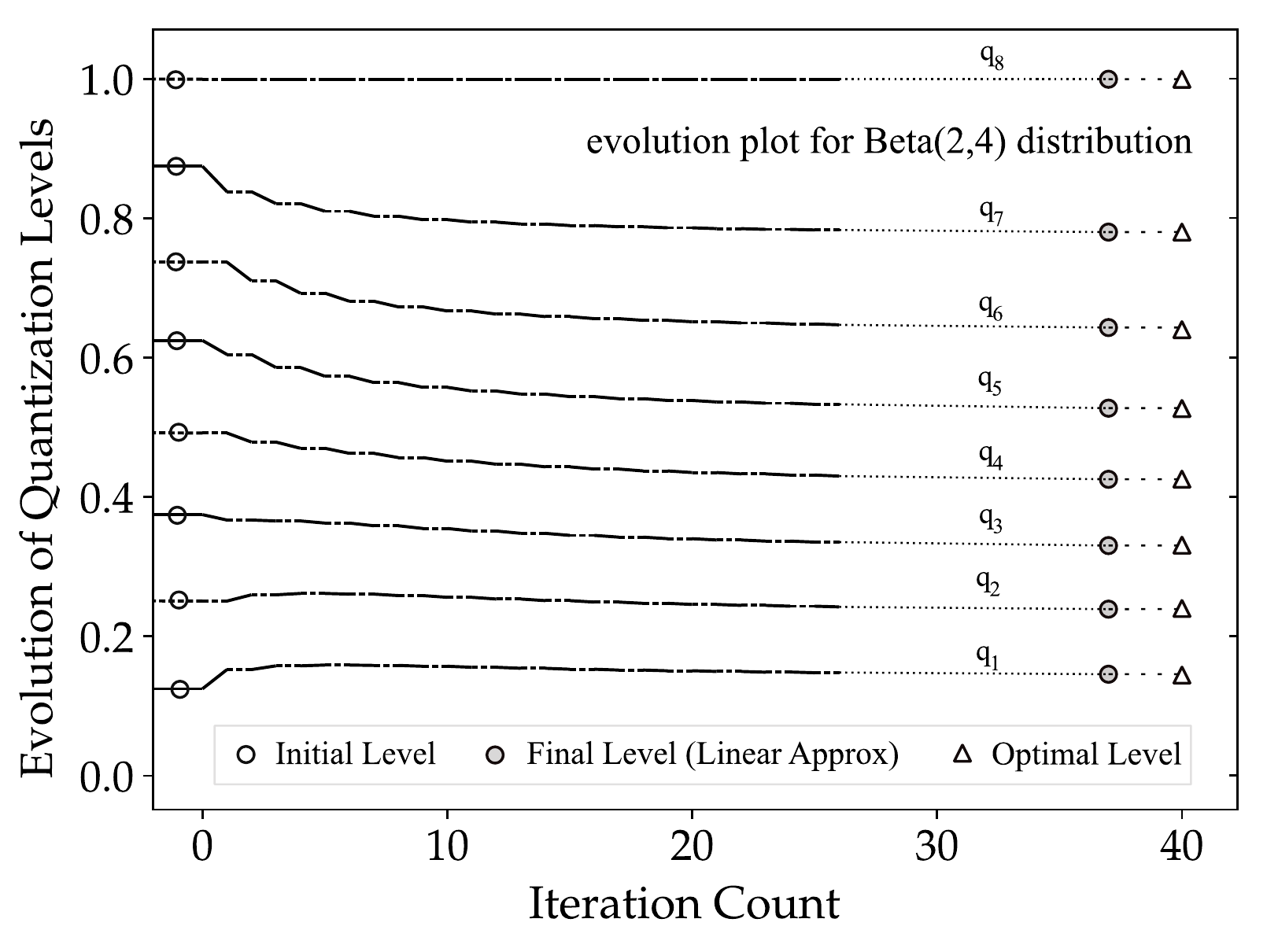}
  \caption{\label{fig:evolution_env_2_4} Evolution of quantization levels for
  the Beta distribution with parameters $\alpha=2$ and $\beta =4$ with a
  equi-spaced initialization. The plot shows the level update according to the
  \emph{Alternating Between Even and Odds (ABEO)} algorithm. The final levels
  due to the linear approximation method are comparable to the optimal (true)
  levels computed without the approximation.} 
\end{figure}
\begin{figure*}
  \centering
  \includegraphics[scale=1.3]{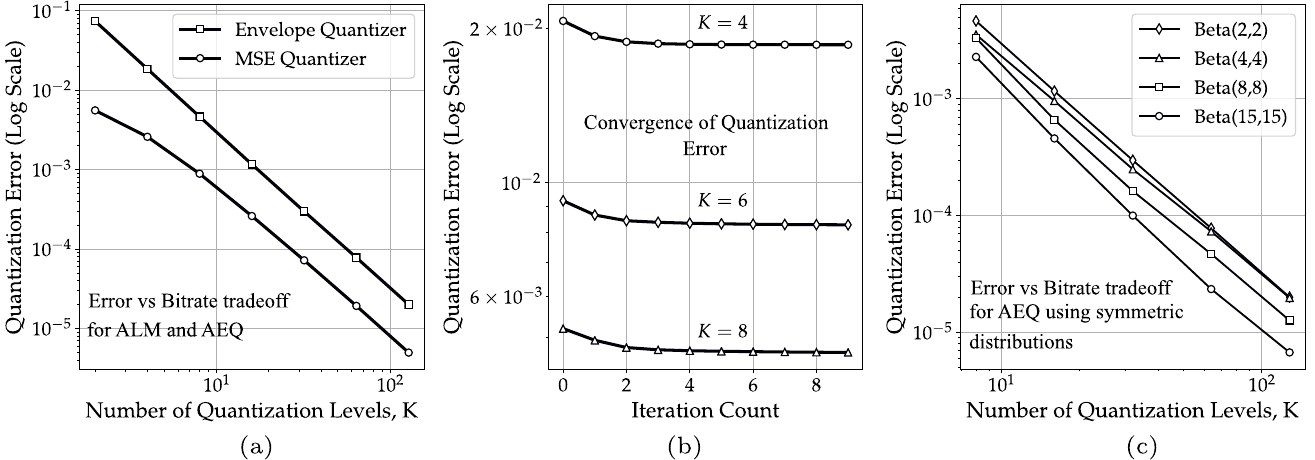}
   \caption{\label{fig:comb_fig1}(a) Quantization error performance
  with the number of quantization levels, $K$ for Approximate Envelope
  Quantizer (AEQ) and Approximate Lloyd-Max (ALM). (b) Convergence of AEQ
 (see Algorithm.~\ref{algo:scalar_quant}) with the increasing
  number of iterations. (c) MSE for envelope quantizer vs number of levels $K$,
  plotted for various symmetric Beta distributions.}
\end{figure*}
\begin{figure*}
  \centering
  \includegraphics[scale=1.3]{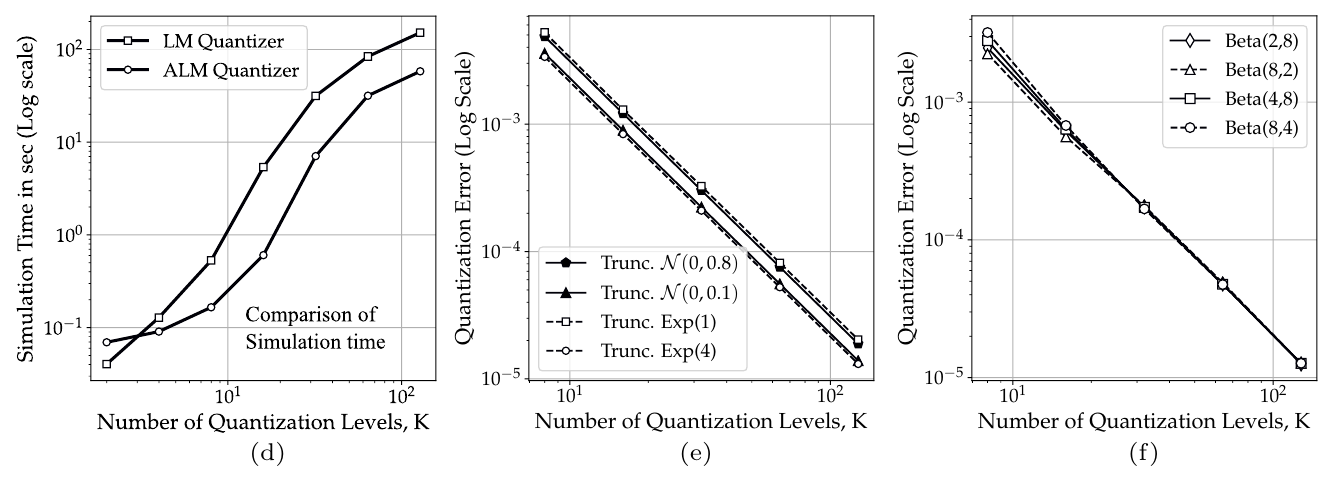}
  \caption{\label{fig:comb_fig2} (d) Simulation time for
  compared for Lloyd-Max quantizer for two cases - with linear approximation and
  without linear approximation. (e) MSE and bitrate tradeoff for AEQ for
  different truncated distributions. (f) Variation of MSE for envelope quantizer
  for different asymmetric Beta distributions.}
\end{figure*}
\begin{figure*}
  \centering
  \includegraphics[scale=1.3]{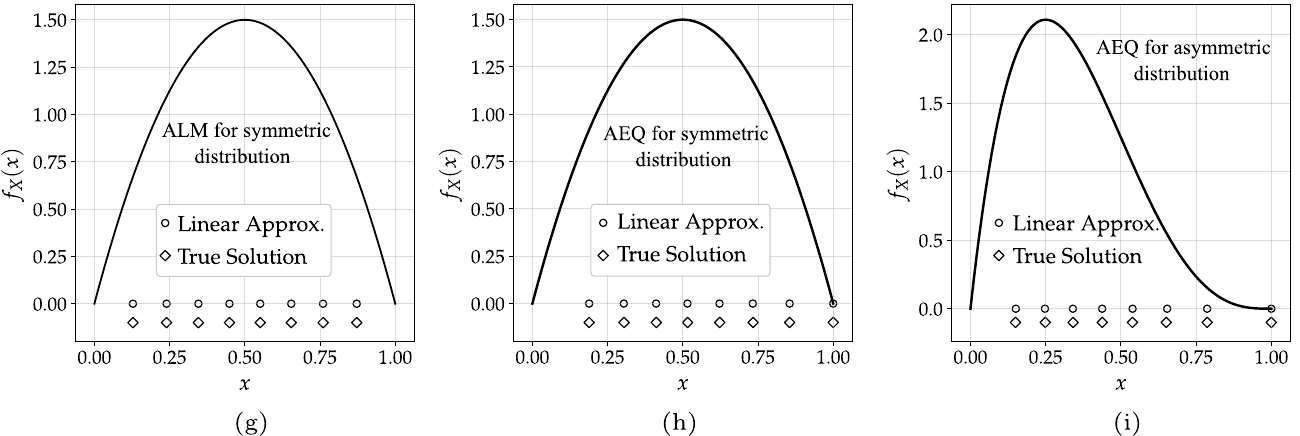}
  \caption{\label{fig:variation_plots} (g) Optimal and Near optimal
  (that is ALM) quantization levels for Lloyd-Max quantizer on Beta(2,2)
  distribution.(h) Optimal and Near optimal (that is AEQ) quantization levels
  for envelope quantizer on Beta(2,2) distribution. (i) Comparison of AEQ near
  optimal scheme with the ground truth optima for the asymmetric Beta(2,4)
  distribution. } 
\end{figure*}
\section{Conclusions} \label{sec:concl}
We have introduced two novel methods for scalar quantization of sources with a
known probability distribution on a finite support. The first quantizer, termed
as ALM, effficiently determines the quantization level updates for the mean
square error cost function. While the second, known as AEQ, is a relatively new
quantization approach, that minimizes an envelope constrained cost function. A
piecewise linear approximation method has been employed here, that significantly
reduces the computational cost of the quantizer. A novel parallel update rule
known as the ABEO, assists in efficiently computing the quantizer levels in a
vectorized manner.  Both our quantization algorithms has been shown to converge
at exponential rate to the unique global minimizer of the cost function. This
convergence result, which is a key contribution of this work, has its roots in
the row-stochastic property associated with the quantization level update
matrices. Simulation results have shown the validity of our analytical results.

\appendices
\section{Proof of properties 1-6 for uniform quantization}
\label{appdx:unif_mat_prop}
\textit{Proof}: 1) Since $P_1$ and $P_2$ are row stochastic, for a fixed $i$
\begin{align}
  \sum_j [P]_{ij} & = \sum_j [P_2 P_1]_{ij} \nonumber \\ 
  & = \sum_j \sum_k [P_2]_{ik} [P_1]_{kj} \nonumber \\
  & = \sum_k [P_2]_{ik} \sum [P_1]_{kj} \nonumber \\ 
  & = \sum_k [P_2]_{ik} \times 1 =1 
\end{align}
This shows $P$ is row stochastic.

2) Consider a vector $\vec{v}$, and 
\begin{align}
  |P\vec{v}|_1 \leq \|P\|_1 \|v\|_1 
\end{align}
Since $\|P\|_1 = 1$, for all eigenvalues $|\lambda| \leq 1$.

3) Follows from 1).

4) $P$ is row symmetric in a cyclic sense, ie $\vec{r}_i = \text{Flip}
(\vec{r}_{n-i+1})$. This follows from row symmetry of individual matrices $P_1$
and $P_2$.  For a row symmetric matrix the eigenvectors satisfy, 
\begin{align}
  P\vec{v}  = \begin{bmatrix} P_u \\ P_l \end{bmatrix} \begin{bmatrix} \vec{v}_u
\\ \vec{v}_l \end{bmatrix} & = \lambda \begin{bmatrix} \vec{v}_u
\\ \vec{v}_l \end{bmatrix} \\
P_u \begin{bmatrix} \vec{v}_u \\ 0 \end{bmatrix}  = \lambda \begin{bmatrix}
  \vec{v}_u \\ 0 \end{bmatrix}, \quad P_l \begin{bmatrix} 0 \\ \vec{v}_l
    \end{bmatrix}  & = \lambda \begin{bmatrix} 0 \\  \vec{v}_l  \end{bmatrix}
\end{align}
By solving the above equations we get $\vec{v}_u = \vec{v}_l$ or $\vec{v}_u = -
\vec{v}_l$.

5) We show that $P-I$ has kernel dimension (or nullity) of $2$. We observe that
the diagonal element $[P-I]_{ii} = 0$ for $i = 2,3,\cdots,n-1$. And using the fact that
each row is shifted with a preceding zero. One can show that all rows are
independent except for the first and the last row. This shows that the geometric
multiplicity is $2$.

6) This follows from 3) and 5).

\section{Existence of a real root for the linear approximation}
\label{appdx:exist_root}
\subsection{Roots corresponding to ALM} \label{subsec:ALM_exist_roots}
We observe that the piecewise linear approximation on the density function
results in the cubic polynomial as described in Table~\ref{tab:ALM_solve_cubic}. Here
we show that there always exists a real root for the polynomial in the interval
$[q_{k-1},q_{k+1}]$. The method employed here uses the fact that a sign change
in the polynomial evaluated between any two points indicates a root between
them. This result is known as the intermediate value theorem. We demonstrate the
steps in the proof, by determining the value of $r(u) = r_0 + r_1 u + r_2 u^2 +
r_3 u^3$ for each of the end points.

\noindent \underline{Case 1:} $2 \leq k \leq K-1$. (Here we note that $r_3= 0$)
\begin{align}
  r(q_{k-1}) & = r_0 + r_1 q_{k-1} + r_2 q_{k-1}^2 \nonumber \\
  & = -\frac18 (q_{k+1}-q_{k-1})^2\left[m_k\left(\frac23 q_{k-1} + \frac13
  q_{k+1}\right)+c\right] \nonumber \\
  & = -\frac18 (q_{k+1}-q_{k-1})^2\left[\frac23 f_X(q_{k-1}) + \frac13
  f_X(q_{k+1})\right] \nonumber \\
  & < 0
\end{align}
At the right boundary, that is, $u = q_{k+1}$,
\begin{align}
  r(q_{k+1}) & = r_0 + r_1 q_{k+1} + r_2 q_{k+1}^2 \nonumber \\
  & = \frac18 (q_{k+1}-q_{k-1})^2\left[\frac23 f_X(q_{k+1}) + \frac13
  f_X(q_{k-1})\right] \nonumber \\
  & > 0
\end{align}
\noindent \underline{Case 2:} $k = 1$ or $K$. We show the proof for $K=1$ and
explain the modifications necessary for $k=K$. The polynomial $r(u)$ evaluated
at $u = q_{0}$ and $u=q_{2}$ simplifies to :
\begin{align}
  r(q_0) &= -\frac12 (q_2-q_0)^2 \left[\frac13 f_X(q_2) + \frac23
  f_X(q_0)\right] \nonumber \\
  & < 0 \\
  r (q_2) &= \frac12 (q_2-q_0)^2 \left[\frac23 f_X(q_2) + \frac13
  f_X(q_0)\right] \nonumber \\
  & > 0.
\end{align}
For the cubic polynomial corresponding to $k=K$, we note that all signs are
reversed with respect to the $k=1$ case. The modified polynomial still evaluates to
give the change in signs at the end points, that is $r(q_{K-1})<0$ and
$r(q_{K'})>0$.
\subsection{Roots corresponding to AEQ}\label{subsec:AEQ_exist_roots}
We show that the polynomial equation $p(u) = p_0 + p_1 u+ p_2 u^2 + p_3 u^3 = 0$, with
coefficients as listed in Table.~\ref{tab:AEQ_solve_cubic}, has atleast one real
root in the interval $[q_{k-1},q_{k+1}]$. Recall that $q_{k-1}$ and $q_{k+1}$
represents the left and right nearest neighbors of the quantization level $q_k$.
We show the above fact using the intermediate value theorem, that is, $p(u)=0$ if
$p(q_{k-1})p(q_{k+1}) < 0$. Evaluating the polynomial at the end points of the
interval we get,
\begin{align}
  p(q_{k-1}) & = p_0 + p_1 q_{k-1} + p_2 q_{k-1}^2 + p_3 q_{k-1}^3 \nonumber \\
  & = -c_k q_{k-1}^2 - c_k q_{k+1}^2 + 2c_k q_{k-1} q_{k+1} \nonumber \\ & \;\;\;\; -m_k q_{k-1}^3 -m_k  q_{k-1}q_{k+1}^2 +2 m_k q_{k-1}^2 q_{k+1} \nonumber \\
  & = -(c_k+m_k q_{k-1}) (q_{k+1}-q_{k-1})^2 \nonumber \\
  & = -f_X(q_{k-1}) (q_{k+1}-q_{k-1})^2 \nonumber \\
  & < 0,
\end{align}
and
\begin{align}
  p(q_{k+1}) & = p_0 + p_1 q_{k-1} + p_2 q_{k-1}^2 + p_3 q_{k-1}^3 \nonumber \\
  & = c_k q_{k-1}^2 + c_k q_{k+1}^2 - 2c_k q_{k-1} q_{k+1} \nonumber \\ &
  \;\;\;\; + \frac23 m_k q_{k-1}^3 + \frac13 m_k  q_{k+1}^3 - m_k q_{k-1}^2 q_{k+1} \nonumber \\
  & = -p(q_{k-1}) + \frac13 m_k (q_{k+1}^3-q_{k-1}^3) \nonumber \\ &\;\;\;\; - m_k q_{k+1} q_{k-1}
  (q_{k+1}-q_{k-1}) \nonumber \\
  & = (\frac23 f_X(q_{k-1})+\frac13 f_X(q_{k+1}))(q_{k+1}-q_{k-1})^2 \nonumber \\
  & > 0.
\end{align}
From the above two inequalities we observe that the product
$p(q_{k-1})p(q_{k+1}))$ is always negative and hence there always exist a root of
$p(u)=0$ in the interval $[q_{k-1},q_{k+1}]$.
\section{Proof that AEQ optimality condition results in a positive
derivative}\label{appdx:AEQ_pos_derivative}

In this section we show that the AEQ optimality condition defined by the
polynomial $p(u)$, in Table.~\ref{tab:AEQ_solve_cubic} has a positive slope. The
proof for the same follows from the convexity of the cost function
\eqref{eq:cost_env}. The derivative of the polynomial $p(q_k)$ with respect to
$q_k$ is given as,
\begin{align}
  \frac{d p(q_k)}{d q_k} & =
  2(q_{k+1}-q_{k-1})\left(m_k\left(\frac{q_{k+1}+q_{k-1}}{2}\right)+c_k\right) &
  \nonumber \\ & \qquad \qquad -2m_k(q_{k+1}-q_{k})^2  \nonumber \\
  & = 2 (q_{k+1}-q_{k-1}) f_{\app}\left(\frac{q_{k+1}+q_{k-1}}{2}\right) 
  \nonumber \\ & \qquad \qquad -2m_k(q_{k+1}-q_{k})^2  \label{eq:p_derivative}
\end{align}
We consider the following three cases - $m_k=0$, $m_k <0$ and $m_k>0$. In the
first and second case, we see that the derivative is positive since
$f_{\app}\left(\frac{q_{k+1}+q_{k-1}}{2}\right) >0$. When $m_k >0$, we use the
fact that, the optimal solution $q_k$ is closer to $q_{k+1}$ than $q_{k-1}$.
In other words, we get the condition $q_{k+1}-q_{k} \leq q_{k} - q_{k-1}$. Using the
above fact, we rewrite \eqref{eq:p_derivative} as,
\begin{align}
  \frac{d p(q_k)}{d q_k} & \geq 2(q_{k+1}-q_{k-1})
  \left[m_k\left(\frac{2q_k-q_{k+1}+q_{k-1}}{2}\right) +c_k\right]
  \nonumber \\
  & \geq 2(q_{k+1}-q_{k-1}) f_{\app}(q_{k-1}) >0
\end{align}

\bibliographystyle{IEEEtran}
\bibliography{refs}



\end{document}